\documentclass[conference ,pt10]{IEEEtran}
\IEEEoverridecommandlockouts
\usepackage{graphicx}
\usepackage{epstopdf}
\usepackage{algorithm}
\usepackage{algorithmic}
\usepackage{subfig}
\usepackage{dcolumn}
\usepackage{epsfig}

\usepackage{float}
\usepackage{amsmath}
\usepackage{amssymb}
\usepackage{amsthm}
\def\BibTeX{{\rm B\kern-.05em{\sc i\kern-.025em b}\kern-.08em
		T\kern-.1667em\lower.7ex\hbox{E}\kern-.125emX}}

\newcommand{\vect}[0]{\text{vect}}
\newcommand{\unvect}[0]{\text{unvect}}
\newcommand{\vectdiag}[0]{\text{vectdiag}}
\newcommand{\rank}[0]{\text{rank}}

\newcommand{\bgamma}[0]{\boldsymbol{\gamma}}
\newcommand{\bGamma}[0]{\boldsymbol{\Gamma}}
\newcommand{\bdelta}[0]{\boldsymbol{\delta}}
\newcommand{\bDelta}[0]{\boldsymbol{\Delta}}

\newcommand{\btheta}[0]{\boldsymbol{\theta}}

\newcommand{\bpsi}[0]{\boldsymbol{\psi}}
\newcommand{\bPsi}[0]{\boldsymbol{\Psi}}

\newcommand{\bA}[0]{\mathbf{A}}

\newcommand{\bB}[0]{\mathbf{B}}

\newcommand{\bC}[0]{\mathbf{C}}

\newcommand{\bD}[0]{\mathbf{D}}

\newcommand{\bE}[0]{\mathbf{E}}

\newcommand{\bG}[0]{\mathbf{G}}

\newcommand{\bI}[0]{\mathbf{I}}

\newcommand{\bJ}[0]{\mathbf{J}}

\newcommand{\bK}[0]{\mathbf{K}}

\newcommand{\bM}[0]{\mathbf{M}}

\newcommand{\bP}[0]{\mathbf{P}}

\newcommand{\bQ}[0]{\mathbf{Q}}
\newcommand{\br}[0]{\mathbf{r}}
\newcommand{\bR}[0]{\mathbf{R}}

\newcommand{\bS}[0]{\mathbf{S}}

\newcommand{\bT}[0]{\mathbf{T}}

\newcommand{\bU}[0]{\mathbf{U}}

\newcommand{\bV}[0]{\mathbf{V}}

\newcommand{\bX}[0]{\mathbf{X}}
\newcommand{\by}[0]{\mathbf{y}}

\newcommand{\bZ}[0]{\mathbf{Z}}

\newcommand{\bGt}[0]{\mathbf{\tilde{G}}}

\newcommand{\bVt}[0]{\mathbf{\tilde{V}}}

\newcommand{\bdeltat}[0]{\boldsymbol{\tilde{\delta}}}
\newcommand{\bGammat}[0]{\boldsymbol{\tilde{\Gamma}}}

\newcommand{\brh}[0]{\mathbf{\hat{r}}}

\newcommand{\bRh}[0]{\mathbf{\hat{R}}}

\newcommand{\zeros}[0]{\mathbf{0}}
\newcommand{\ones}[0]{\mathbf{1}}

\newcommand{\MCC}[0]{\mathcal{C}}

\newcommand{\MCE}[0]{\mathcal{E}}

\newcommand{\MCN}[0]{\mathcal{N}}

\newcommand{\MCbV}[0]{\boldsymbol{\mathcal{V}}}

\newcommand{\beq}{\begin{equation}}
\newcommand{\eeq}{\end{equation}}
\newcommand{\bea}{\begin{array}}
\newcommand{\ena}{\end{array}}

\newcommand{\DL}{\begin{dashlist}}
\newcommand{\DLE}{\end{dashlist}}

\newtheorem{theorem}{Theorem}

\title{Constrained Least Squares for Extended Complex Factor Analysis\\ \thanks{$1$ A. Mouri Sardarabadi (ammsa@astro.rug.nl) and L.V.E. Koopmans are affiliated with Kapteyn Astronomical Institute, University of Groningen, The Netherlands. $2$ Alle-Jan van der Veen is affiliated with Delft University of Technology, Delft, The Netherlands}}
\author{Ahmad Mouri Sardarabadi $^1$,  Alle-Jan van der Veen$^2$ and L.V.E. Koopmans$^1$}
\newcommand{\bJA}[0]{\bJ_{\bA}}
\begin{document}
\maketitle
\begin{abstract}
For subspace estimation with an unknown colored noise, Factor Analysis (FA) is a good  candidate for replacing the popular eigenvalue decomposition (EVD). Finding the unknowns in factor analysis can be done by solving a non-linear least square problem. For this type of optimization problems, the Gauss-Newton (GN) algorithm is a powerful and simple method. The most expensive part of the GN algorithm is finding the direction of descent by solving a system of equations at each iteration. In this paper we show that for FA, the matrices involved in solving these systems of equations can be diagonalized in a closed form fashion and the solution can be found in a computationally efficient way. We show how the unknown parameters can be updated without actually constructing these matrices. The convergence performance of the algorithm is studied via numerical simulations.
\end{abstract}

\section{Introduction}
The eigenvalue decomposition (EVD) of the data covariance
matrix (a.k.a. principal component analysis) is a powerful tool for subspace based and low-rank approximation techniques.  Without noise, the data covariance matrix is considered to be rank-deficient,
and its column span is called the signal subspace.  In the presence of additive noise
an implicit assumption is that this noise is white. If this is not the case but the noise covariance matrix is known from calibration, whitening techniques can be used as 
a pre-processing step.  However, in many array processing
applications this knowledge is not available. A preferable
approach is to replace the EVD by techniques that jointly estimate
the signal subspace and the noise covariance matrix. Using FA as a substitute for EVD in these cases was suggested by \cite{aj05skabook} and \cite{mouri2018}. The latter also includes several extension of the classical FA model which were denoted as Extended Factor Analysis (EFA).

The FA model was introduced by Spearman \cite{Spearman1904} in 1904 and became an established multivariate technique mainly due to the work of Lawley, Anderson, J\"{o}reskog
and others between 1940 and 1970 \cite{lawley1940vi2,AndersonRubin,joreskog1969,lawley71,mardia79}. The FA model
assumes a covariance matrix $\bR$ of the data under study (e.g.,
samples acquired from an array of sensors) can be modeled as
\begin{equation}
\label{eq:famodel}
\bR = \bA\bA^H + \bD,
\end{equation}
where $^H$ is the Hermitian transpose, $\bA$ is a ``tall'' matrix ($\bA\bA^H$ has low rank), and
$\bD$ is a positive diagonal matrix.  In terms of subspace-based
techniques, $\bA$ captures the signal subspace while $\bD$ can model
the noise covariance matrix.  Given a sample covariance matrix
$\bRh$, the objective of FA is to estimate $\bA$ and $\bD$.

In this paper we use the EFA which was introduced by \cite{mouri2012, mouri2018} and present a new algorithm based on the Gauss-Newton (GN) method. By showing that a closed-form Singular Valude Decomposition (SVD) of the Jacobian matrix can be found, we develop an efficient algorithm that has competitive complexity and performance.

The structure of the paper is as follows: in Sec.~\ref{sec:datamodel} we discuss the EFA model, in Sec.~\ref{sec:LSEFA} we introduce the least squares (LS) formulation of the problem and discuss our method and in Sec.~\ref{sec:simulations} we use numerical simulations to evaluate the performance of the algorithm.

\section{Data Model}
\label{sec:datamodel}
We assume to have access to $N$ samples from the outputs of $P$ receivers such that a sample covariance matrix can be constructed using
\begin{equation}
\bRh = \frac{1}{N} \sum_{n=1}^N \by[n]\by[n]^H
\end{equation}
where the $P \times 1$ vector $\by[n]$ represents the zero mean sampled output of the receivers. We also assume that the Extended Factor Analysis \cite{mouri2018} model is valid
\begin{equation}
\bR = \MCE\{\bRh\} = \bA\bA^H + \bPsi,
\end{equation}
where $\MCE\{.\}$ represents the expectation of the argument, $\bA$ is a $P \times Q$ matrix with $\rank(\bA)=Q$ and $\bPsi$ is a structured matrix such that \mbox{$\bPsi = \bM \odot \bPsi$} for some mask matrix $\bM$ consisting of zeros and ones. We assume to know the structure of $\bPsi$ (and hence $\bM$) in advance, and define a selection matrix, $\bS$, for the non-zero elements of $\bM$. Using this selection matrix we can stack the non-zero elements of $\bPsi$ as \mbox{$\bpsi = \bS^T \vect(\bPsi)$} and similarly \mbox{$\vect(\bPsi) = \bS\bpsi$}. For the classical FA, $\bM = \bI$ and $\bPsi = \bD$ is diagonal. Given a sample covariance matrix we are interested in finding $\bA$ and $\bPsi$ or only one of them, depending on the application.

This parametrization of the covariance model is not unique and for any unitary matrix $\bQ$, $\bA\bQ$ results in exactly the same model. As a result, in the complex case, the problem needs additional $Q^2$ constraints for a unique parametrization \cite{aj05skabook,mouri2018}. There are several ways to constrain $\bA$ such as restricting it to be lower triangular with real (positive) diagonal elements, or forcing its columns to be orthogonal such that \mbox{$\bA^H\bA = \bGamma$} is a diagonal matrix. Regardless of the chosen constraints, the total number of unique parameters for the problem is
\begin{equation}
\label{eq:nhat}
\hat{n} = 2PQ + \|\bM\|_1 - Q^2.
\end{equation}

We show that using the constraints of the form \mbox{$\bA^H\bA = \bGamma$} can be advantageous when solving the EFA model using the least squares cost function. In the next section we present our approach. 

\section{Least Squares for EFA}
\label{sec:LSEFA}
\subsection{Gauss-Newton Method}
The aim is to find the FA model for a sample covariance $\bRh$ matrix by minimizing the following non-linear LS cost function:
\[
\min_{\bA,\bPsi} \|\bRh-\bA\bA^H-\bPsi\|_F^2.
\]
Given the non-linearity of the cost function we use the Gauss-Newton method to find the solution similar to \cite{lee1978,mouri2018}. The GN algorithm for solving this problem consists of the following updates
\[
\btheta^{(k+1)}=\btheta^{(k)}+\mu_k \bdelta
\]
where $\bdelta$ is the direction of descent and the solution to the following system of linear equations
\begin{equation}
\label{eq:sovedelta}
\bJ^H\bJ\bdelta=\bJ^H(\brh-\br),
\end{equation}
where $\vect(.)$ is the vectorization operator which stacks the columns of the argument matrix into a single column vector, $\brh=\vect(\bRh)$, $\br=\vect(\bA\bA^H+\bD)$, and \mbox{$\bJ={\partial \br}/{\partial \btheta^T}$} is the Jacobian matrix. Using the relation between the Kronecker product and vectorization, we find the Jacobian matrix to be
\[
\bJ=\frac{\partial \br}{\partial \btheta^T}=\begin{bmatrix}\bA^*\otimes \bI_P & (\bI_P \otimes \bA)\bK^{P,Q} & \bS \end{bmatrix},
\]
where 
\[
\btheta=\begin{bmatrix}
\vect(\bA)^T &
\vect(\bA)^H &
\bpsi^T
\end{bmatrix}^T,
\]
$\bK^{P,Q}$ is a permutation such that \mbox{$\vect(\bX^T) =\bK^{P,Q} \vect(\bX)$} for any $P \times Q$ matrix $\bX$, $\otimes$ represents the Kronecker product and $^*$~represent the complex conjugate. We also have \mbox{$\bK^{P,Q} \bK^{Q,P} = \bI$}.

The most expensive operation during the GN updates is the solution of \eqref{eq:sovedelta}. We also know that the non-uniqueness of $\bA$ causes the Jacobian to be rank deficient by $Q^2$, which makes the system singular. A closed form solution of this system was first introduced in \cite{mouri2018} using a block LDU decomposition. While that method is more general and also works for weighted least squares (WLS), it does not produce a minimum length solution and also does not clearly show how \eqref{eq:sovedelta} is reduced to a non-singular system. In the next section we show that a unique minimum length solution can be found when the constraint $\bA^H\bA$ is enforce at each iteration.

\subsection{Model Reduction}
In this section we show how the constraints \mbox{$\bA^H\bA = \bGamma$} can be used to transform the problem into a reduced one which does not suffer from singularities. We also show how this transformation can be used to solve the Gauss-Newton system of equations in an efficient way. One way to find $\bA$ that satisfies the constraints is to use the singular value decomposition. Let the SVD of $\bA$ be
\begin{equation}
\label{eq:svdA}
\bA=\left [ \begin{array}{c|c} \bU_0 & \bU_n \end{array} \right ] \left [\begin{array}{c} \bGamma \\ \hline \zeros \end{array} \right ]^{1/2} \bQ^H.
\end{equation}
If $\bA$ already satisfies the constraints then $\bQ=\bI_Q$. Otherwise, we can always multiply $\bA$ from right by $\bQ$ to make sure that the constraints are satisfied. For simplicity we also define 
\[
\bJA = \begin{bmatrix}\bA^*\otimes \bI_P & (\bI_P \otimes \bA)\bK^{P,Q} \end{bmatrix},
\]
which is the submatrix of $\bJ$ correspoding to derivatives with respect to $\bA$ and $\bA^*$.
\begin{theorem}
	\label{theorem:evdJAHJA}
	Given the SVD of $\bA$ as defined by \eqref{eq:svdA}, the eigenvalue decomposition of $\bJA^H\bJA$ is given by
	\begin{equation}
	\label{eq:EVDJAHJA}
	\bJA^H\bJA = \bVt \bGammat \bVt^H,
	\end{equation}
	where 
	\[
		\bGammat = \left [\begin{array}{c c c | c}
		\bGamma \otimes \bI & & &  \\
		&\bI \otimes \bGamma & &  \\
		&& \bGamma \otimes \bI + \bI \otimes \bGamma &\\
		\hline
		&&& \zeros
		\end{array} \right],
	\]
	is a diagonal matrix containing the eigenvalues and $\bVt$ is a unitary matrix containing the eigenvectors. The matrix $\bVt$ can be defined as  \mbox{$\bVt = [\bV|\bZ]$} where
	\[
	\begin{array}{l}
	\bV= \\
	\begin{bmatrix} 
	\bI_{PQ} & \zeros \\
	\zeros& \bK^{Q,P} 
	\end{bmatrix}
	\begin{bmatrix}
	\bI_Q \otimes \bU_n &  \zeros & (\bGamma^{1/2} \otimes \bU_0)\bG^{1/2} \\ 
	\zeros & (\bU_n^* \otimes \bI_Q)   & (\bU_0^* \otimes \bGamma^{1/2})\bG^{1/2}
	\end{bmatrix},
	\end{array}
	\]
	and
	\begin{equation*}
	\begin{array}{l}
		\bZ = \begin{bmatrix} (\bI_Q \otimes \bA) \\ -\bK^{Q,P}(\bA^* \otimes \bI_Q)\end{bmatrix}\bG^{1/2}.
	\end{array}
	\end{equation*}
	with \mbox{$\bG=(\bGamma \otimes \bI_Q +  \bI_Q \otimes \bGamma)^{-1}$}.
\end{theorem}
\begin{proof}
	Carrying out the multiplications on both sides of \eqref{eq:EVDJAHJA} verifies the equality, but it is important to show that $\bVt$ is a unitary matrix. 
	It is trivial to show that $\bV^H\bV=\bI$ using the fact that $\bU_n^H\bU_n=\bI_{P-Q}$, $\bU_0^H\bU_0=\bI_Q$, $\bU_n^H\bU_0=\zeros$ and the property $(\bA\otimes \bB)^H(\bC\otimes \bD)=(\bA^H\bC\otimes \bB^H\bD)$. What is left to show is that both $\bJA\bZ$ and $\bV^H\bZ$ are zero.
	
	We can easily verify that 
	\[
	\bJA\bZ=[(\bA^* \otimes \bA)-(\bA^* \otimes \bA)]\bG=\zeros,
	\]
	This shows that $\bZ$ is in the null-space of $\bJ$. We also have 
	\begin{align*}
	\bZ^H\bZ&=\bG^{1/2}[(\bA^T\bA^* \otimes \bI_Q)+(\bI_Q \otimes \bA^H\bA)]\bG^{1/2}\\
	& =  \bG^{1/2}[(\bGamma \otimes \bI_Q)+(\bI_Q \otimes \bGamma)]\bG^{1/2} = \bI_{Q^2}
	\end{align*}
	which confirms that $\rank(\bZ)=Q^2$. Now that we have shown that $\bZ$ is a basis for the null-space of $\bJ$, it is sufficient to show that $\bV^H\bZ=\zeros$. Again, using $\bU_n^H\bA=\zeros$ and $\bU_0^H\bA=\bGamma^{1/2}$ we can easily verify that $\bV^H\bZ=\zeros$.
\end{proof}
Using this result we can define a new basis 
\[
\MCbV = \begin{bmatrix}
\bV & \zeros \\
\zeros & \bI
\end{bmatrix},
\]
which allows us to remove the null-space of $\bJ$ and reformulate the system of linear equations for solving the direction of decent $\bdelta$ as
\[
\begin{array}{r l}
\bJ^H\bJ\bdelta&=\bJ^H(\brh-\br)\\
\MCbV\MCbV^H\bJ^H\bJ\MCbV\MCbV^H\bdelta&=\MCbV\MCbV^H\bJ^H(\brh-\br)\\
(\MCbV^H\bJ^H\bJ\MCbV)\bdeltat&=\MCbV^H\bJ^H(\brh-\br),
\end{array}
\]
where in the last equation we changed the variable \mbox{$\bdeltat \equiv  \MCbV^H\bdelta$}. Because the systems are consistent we have $\bdelta = \MCbV\bdeltat$. In the next section we show how to find the solution to this system of equations without actually constructing the matrices $\bJ$ and $\MCbV$.
	\begin{figure*}[h!]
	\hspace*{-.5cm}
	\centering
	\subfloat[Convergence speed for $N = 1000$, $Q = 20$\label{fig:fig:conv_speed_p100_q20_N1000}]{%
		\includegraphics[width = 0.48\textwidth]{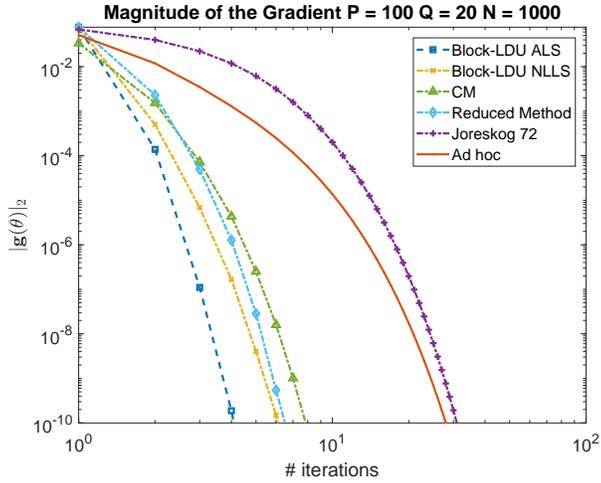}
	}
	\hfill
	%\hspace*{-0.5cm}
	\subfloat[Convergence speed for $N = 1000$, $Q = 80$\label{fig:conv_speed_p100_q80_N1000}]{%
		\includegraphics[width = 0.48\textwidth]{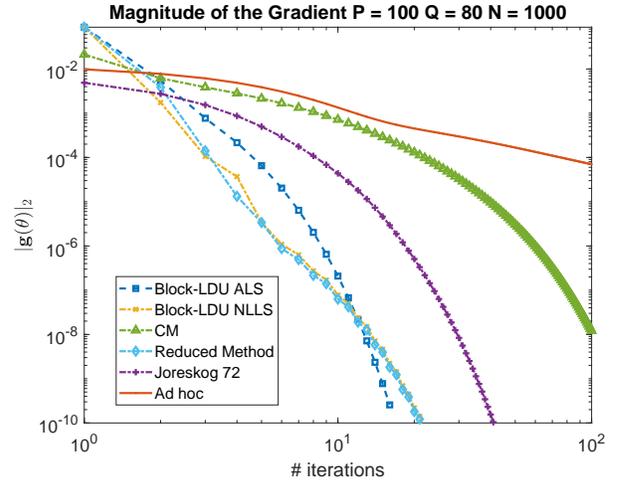}
	}
	
	\vfill
	%\hspace*{-0.5cm}
	\subfloat[Faild Convergence for different noise realization of the same model.\label{fig:convv_speed_p100_q80_N1000}]{%
		\includegraphics[width = 0.48\textwidth]{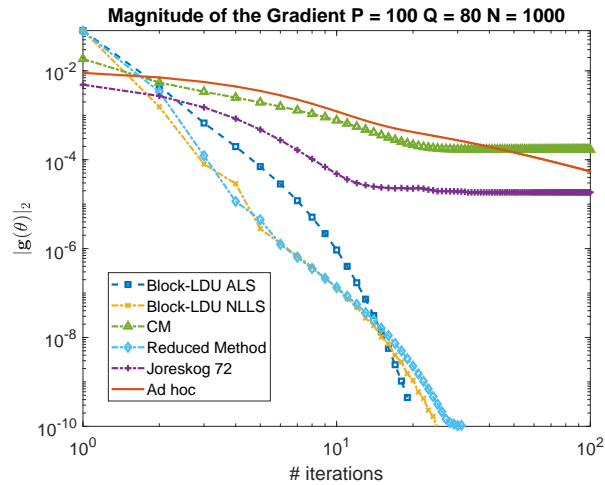}
	}
	\caption{Convergence for $P = 100$ sensors and varying number of
	samples $N$ and sources $Q$.}
	\label{fig:conv1}
\end{figure*}
\subsection{Improved Linear Solver}
In \cite{mouri2018} it was shown that we can split the direction of descent, $\bdelta$ as
\[
\bdelta = [\bdelta_{\bA}^T, \bdelta_{\bPsi}]^T,
\]
and find each sub-vector sequentially. This can be done because  $\bJA$ and $\bS$ must be linearly independent for the (E)FA problem to be identifiable \cite{mouri2018}. First, using the generalized matrix inversion-lemma \cite{rohde1965generalized} we find $\bdelta_{\bPsi}$ by solving 
\begin{equation}
\label{eq:DeltaPsi}
(\bI - \bT) \bdelta_{\bPsi} = \bS^H\vect[\bP^\bot (\bRh - \bPsi) \bP^\bot],
\end{equation}
where \mbox{$\bT = \bS^H(\bP^T \otimes \bI + \bI \otimes \bP - \bP^T \otimes \bP)\bS$}, \mbox{$\bP = \bU_0\bU_0^H$} and \mbox{$\bP^\bot = \bI - \bP$}. The matrices $\bI - \bT$ and $\bT$ are positive definite and positive semi-definite respectively. This means that $\bT$ is convergent, i.e. its spectral radius satisfies $\rho(\bT) < 1$. This allows us to use simple iterative solvers to find the solution in an efficient way.
Once we have $\bdelta_{\bPsi}$, we find $\bdelta_{\bA}$ by solving
\[
(\bV^H\bJA^H\bJA\bV)\bdeltat_{\bA}=\bV^H\bJA^H(\brh-\br-\bS\bdelta_{\bPsi}),
\]
followed by $\bdelta_{\bA} = \bV\bdeltat_{\bA}$. Using the results from Theorem \ref{theorem:evdJAHJA} we have
\[
(\bV^H\bJA^H\bJA\bV)^{-1} = \begin{bmatrix}
\bGamma^{-1} \otimes \bI & & \\
& \bI \otimes \bGamma^{-1} &\\
& & \bG\\
\end{bmatrix},
\]
which makes finding $\bdeltat_{\bA}$ trivial. 
Similar to the the definition of $\btheta$, for $\bdelta_{\bA}$ we have
\[
\bdelta_{\bA} = \begin{bmatrix}
\vect(\bDelta_{\bA})\\
\vect(\bDelta_{\bA}^*) 
\end{bmatrix}.
\]
The matrix form of the solution, $\bDelta_{\bA}$, can be found by introducing some intermediate results:
\begin{align}
\label{eq:deltaA}
\bE &= \bRh - \bA\bA^H - \bPsi - \bDelta_{\bPsi}\notag \\
\bDelta_1 &= \bU_0[\bGt \odot (\bGamma \bU_0^H \bE \bU_0 + \bU_0^H \bE \bU_0 \bGamma)\bGamma^{1/2}]\notag \\
\bDelta_2 &= (\bI - \bP) \bE \bU_0 \bGamma^{-1/2}\notag \\
\bDelta_{\bA} &= \bDelta_1 + \bDelta_2,
\end{align}
where \mbox{$\bdelta_{\bPsi} \equiv \vect(\bDelta_{\Psi})$} is the matrix form of the direction of descent for $\bPsi$ and $\bGt=\unvect(\vectdiag(\bG^2))$\footnote{Let $\bgamma = \vectdiag(\bGamma)$ then $\bGt = (\bgamma\ones_Q^T + \ones_Q \bgamma^T)^{\odot -2}$ where $^{\odot -1}$ is the Hadamard or element-wise inversion of  a matrix (i.e. for $\bX=[x_{ij}]$, \mbox{$\bX^{\odot -n}=[1/x_{ij}^n]$}).} is a $Q \times Q$ matrix constructed by unvectorizing the diagonal elements of $\bG^2$. By closer inspection we see that $\bDelta_1$ is the update in the current subspace of $\bA$ and $\bDelta_2$ is the update in its null-space. It is also clear that we do not need to explicitly find $\bU_n$ which allows us to use the economic-size SVD rather than a full one (or use EVD of $\bA^H\bA$).

To summarize, we use the following updates during the GN iterations:
\begin{align*}
\bPsi^{(k+1)} = \bPsi^{(k)} + \mu_k \bDelta_{\bPsi}\\
\bA^{(k+1)} = \bA^{(k)} + \mu_k \bDelta_{\bA} 
\end{align*}
where $\bDelta_{\Psi}$ and $\bDelta_{\bA}$ are given by \eqref{eq:DeltaPsi} and \eqref{eq:deltaA} respectively. Once the direction of descent is found we can find the optimal $\mu_k$ using the procedure described in \cite{mouri2018} which requires solving for the roots of a third order polynomial with real coefficients for which closed form solutions exist.

\section{Simulation}
\label{sec:simulations}

We evaluate the convergence speed of various
algorithms for the classical factor analysis model. An array with $P = 100$ elements is simulated.  The matrix $\bA$ is
chosen randomly with a standard complex Gaussian distribution (i.e.\
each element is distributed as $\MCC\MCN(0,1)$) and $\bD$ is chosen randomly
with a uniform distribution between $1$ and $5$.  

For $P=100$, the maximum number of sources is \mbox{$Q_{\max} = 89$}.
We show simulation results for $Q = 20$, representative for low-rank
cases, and for \mbox{$Q = 80$} for high-rank cases.
Sources and noise are generated using standard unit power
complex Gaussian distributions.

We compare our algorithm (denoted as ``Reduced method") with the classical ad-hoc alternating LS (see \cite{anderson2003,mardia79,mouri2018} for details), the Block LDU based solvers proposed in \cite{mouri2018}, J\"{o}reskog
\cite{joreskog1972}, and the more recent CM method by \cite{zhao08}. Because the method presented here is based on LS we set the weighting matrix for the Block LDU based algorithms equal to identity. Fig.~\ref{fig:conv1} shows the convergence behavior based on the gradient.

The proposed algorithm has similar convergence properties as the Block LDU. This is expected because these algorithms also use the GN updates. CM and the ad-hoc method are more sensitive to the number of sources, $Q$, and converge very slowly for large $Q$. Also, for large $Q$, both J\"{o}reskog and CM might fail to converge for the same model but different noise realization see Fig.~\ref{fig:convv_speed_p100_q80_N1000}. While it is possible to find models for which our algorithm also fails to converge as the problem is non-convex and non-linear it is our experience that in those scenarios the other methods also fail to converge with high probability. Also for the same model (i.e. the same $\bA$ and $\bD$) our method is less sensitive to the noise realization.
\section{Conclusion}
In this paper we have shown how constraining the columns of $\bA$ to be orthogonal in a(n) (E)FA problem  allows us to find a diagonalization of the Gauss-Newton approximation of the Hessian. Using this diagonalization we have developed a new non-linear least squares algorithm with competitive convergence speed. We have also shown that the system of equations needed for updating the noise covariance matrix, $\bPsi$, has a convergent matrix which allows fast iterative algorithms to be applied.

Using simulation we demonstrated the newly proposed algorithm has a competitive convergence properties and is robust to different noise realizations.

The proposed algorithm does not guarantee the positive (semi-)definiteness of $\bPsi$. This will be address in future works.

\bibliographystyle{IEEEtran}
\bibliography{biblio3}
\end{document}